\documentclass[a4paper,UKenglish]{lipics}

\newif\ifFullyDynamic
\FullyDynamictrue

\newif\ifShort
\Shortfalse

\usepackage{amssymb,amsmath}
\usepackage{graphicx}
\usepackage{xspace}
\usepackage{wrapfig}
\usepackage{comment}

\usepackage{tikz}
\usetikzlibrary{calc,intersections,through,backgrounds}
\usetikzlibrary{arrows, decorations.pathreplacing}

\newcommand{\mytitle}{Dynamic Conflict-Free Colorings in the Plane\footnote{MdB and AM are supported by the Netherlands' Organisation for Scientific Research (NWO) under project no.~024.002.003.}}
\newcommand{\mytitleshort}{\mytitle}

\title{\mytitle}
\titlerunning{\mytitleshort} 

\newcommand{\reals}{\mathbb{R}\xspace}
\newcommand{\Reals}{\reals}
\newcommand{\naturals}{\mathbb{N}\xspace}
\newcommand{\etal}{\emph{et~al.}\xspace}

\newcommand{\col}{\mathit{col}}
\newcommand{\cf}{\gamma_{\mathrm{um}}}
\newcommand{\C}{\mathcal{C}}    

\newcommand{\objects}{S\xspace}

\newcommand{\ob}{s}
\newcommand{\qu}{Q}    
\newcommand{\sq}{s}    
\newcommand{\family}{\mathcal{F}\xspace}

\newcommand{\claim}[2]{\vspace{6pt} \noindent $\triangleright$ {\sffamily Claim.}{\it #1} \\[3pt]
{\sffamily Proof.} #2 \hfill $\triangleleft$ \vspace{6pt}}

\newcommand{\fact}[1]{\vspace{6pt} \noindent $\triangleright$ {\sffamily Fact. }{\it #1} \\[3pt]}

\newcommand{\mypara}[1]{\vspace{10pt} \noindent \textbf{\sffamily #1}}

\renewcommand{\leq}{\leqslant}
\renewcommand{\geq}{\geqslant}
\newcommand{\tree}{\mathcal{T}}
\DeclareMathOperator{\height}{height}
\newcommand{\myroot}{\mathit{root}}

\newcommand{\eps}{\varepsilon}

\newcommand{\myleft}{\mathit{left}}
\newcommand{\myright}{\mathit{right}}

\newcommand{\myne}{\mbox{\scriptsize\sc ne}}
\newcommand{\myse}{\mbox{\scriptsize\sc se}}
\newcommand{\mynw}{\mbox{\scriptsize\sc nw}}
\newcommand{\mysw}{\mbox{\scriptsize\sc sw}}

\newcommand{\tleftv}{\tree_{\myleft(v)}}
\newcommand{\trightv}{\tree_{\myright(v)}}
\newcommand{\tleftw}{\tree_{\myleft(w)}}
\newcommand{\trightw}{\tree_{\myright(w)}}
\newcommand{\rmax}{r_{\max}}

\newcommand{\smax}{\sq_{\max}}
\newcommand{\smin}{\sq_{\min}}

\makeatletter
\newcommand\niton{\mathrel{\m@th\mathpalette\canc@l\owns}}
\newcommand\canc@l[2]{{\ooalign{$\hfil#1/\mkern1mu\hfil$\crcr$#1#2$}}}
\makeatother

\newif\ifcomments
\commentstrue
\ifcomments
    \newcommand{\mdb}[1]{\textcolor{blue}{[Mark: #1]}}
    \newcommand{\am}[1]{\textcolor{red}{[Aleks: #1]}}
\else
    \newcommand{\mdb}[1]{}
    \newcommand{\am}[1]{}
\fi

\definecolor{nicegreen}{rgb}{0,0.7,0.3}
\definecolor{niceyellow}{rgb}{0.9,0.7,0.07}
\definecolor{nicepurple}{rgb}{0.5,0.2,0.8}

\graphicspath{{./figures/}}
\ifShort 
\else 
\fi

\author[1]{Mark de Berg}
\author[1]{Aleksandar Markovic}
\affil[1]{TU Eindhoven, the Netherlands}
\authorrunning{M.\, de Berg, A.\, Markovic} 

\Copyright{Mark de Berg, Aleksandar Markovic}


\begin{document}

\maketitle



\begin{abstract}
We study dynamic conflict-free colorings in the plane,
where the goal is to maintain a conflict-free coloring
(CF-coloring for short) under insertions and deletions.
\begin{itemize}
\item First we consider CF-colorings of a set~$\objects$ of
      unit squares with respect to points. Our method maintains a CF-coloring
      that uses $O(\log n)$ colors at any time, where $n$ is the current
      number of squares in~$\objects$, at the cost of only $O(\log n)$
      recolorings per \ifShort\else
      insertion or deletion of a square.\fi
      We generalize the method to rectangles whose sides have lengths
      in the range $[1,c]$, where $c$ is a fixed constant. Here the number of used colors
      becomes $O(\log^2 n)$.
      The method also extends to arbitrary rectangles whose coordinates come
      from a fixed universe of size $N$, yielding $O(\log^2 N \log^2 n)$
      colors. The number of recolorings for both methods stays in~$O(\log n)$.
\item We then present a general framework to maintain a CF-coloring under insertions
      for sets of objects that admit a unimax coloring with a small number of colors in the static case.
      As an application we show how to maintain a CF-coloring with $O(\log^3 n)$ colors
      for disks (or other objects with linear union complexity)
      with respect to points at the cost of $O(\log n)$ recolorings per insertion.
      \ifFullyDynamic
        We extend the framework to the fully-dynamic case when the static unimax coloring
        admits weak deletions. As an application we show how to maintain a CF-coloring
        with $O(\sqrt{n} \log^2 n)$ colors for points with respect to rectangles,
        at the cost of $O(\log n)$ recolorings per insertion
        and $O(1)$ recolorings per deletion.
      \fi
\end{itemize}
These are the first results on fully-dynamic CF-colorings in the plane,
and the first results for semi-dynamic CF-colorings for non-congruent
objects.
\end{abstract}


\section{Introduction}
Consider a set of base stations in the plane that can be used for mobile communication.
To ensure a good coverage, the base stations are typically positioned in such a way
that the communication ranges of different base stations overlap.
However, if a user is within range of several base stations using the
same frequency, then interference occurs and the communication
is lost. Therefore, we want to assign frequencies to the base
stations such that any user within range of at least one base station,
is also within range of at least one base station
using \ifShort\else
a frequency where no interference occurs.\fi The easy solution would be
to give all stations a different frequency. However, this is
undesirable as the set of available frequencies is limited.
The question then arises: how many different frequencies are needed
to ensure that any user that is within range of at least one base station
has an interference-free base station at his disposal?
Motivated by this and other applications, Even \etal \cite{even-cf-03}
and Smorodinsky \cite{thesis-smorodinsky} introduced the notion of
\emph{conflict-free colorings} or \emph{CF-colorings} for short.
Here the ranges of the base stations are modeled as regions (disks, or other objects)
in the plane, and \ifShort \else different\fi frequencies are
represented by \ifShort \else different\fi colors.
A CF-coloring is now defined as follows.

Let~$\objects$ be a set of objects in the plane.
For a point $q\in \Reals^2$, let $\objects_q := \{ S \in \objects | q \in S \}$
be the subset of objects containing~$q$.
A coloring $\col:\objects \to \naturals$ of the objects in~$\objects$---here
we identify colors with non-negative integers---is said to be
\emph{conflict-free (with respect to points)}
if for each point~$q$ with $\objects_q\neq \emptyset$ there is an object
$S\in \objects_q$ whose color is unique among the objects in~$\objects_q$.
A CF-coloring is called \emph{unimax} when the maximum color
in $\objects_q$ is unique.

We can also consider a dual version of planar CF-colorings.
Here we are given a set $\objects$ of points and a family $\family$ of geometric ranges,
and the goal is to color the points in $\objects$ such
that any range from~$\family$ containing a least one point, contains
a point with a unique color. \ifShort \else
Both versions of CF-colorings in the plane---coloring
objects with respect to points, and coloring points with respect to ranges---can
be formulated as coloring nodes in a hypergraph such that any hyperedge
has a node with a unique color. In this paper we stick
to the more intuitive geometric view. \fi
\medskip

Conflict-free colorings have received
a lot of attention since they were introduced by
Even \etal \cite{even-cf-03} and Smorodinsky \cite{thesis-smorodinsky};
see the overview paper by Smorodinsky \cite{S-survey-10}, which surveys
the work up to 2010. We review the work most relevant to our results.

Even \etal proved that it is always possible to CF-color
a set of disks in the plane using~$O (\log n)$ colors, and
that~$\Omega (\log n)$ colors are needed in the worst case.
The authors extended the result to sets of translates
of any given centrally symmetric polygon. Later, Har-Peled
and Smorodinsky \cite{harpeled-cf-05} further generalized the
result to regions with near-linear union complexity. The
dual version of the problem was also studied by Even \etal
\cite{even-cf-03}; they showed it is possible
to CF-color points using~$O(\log n)$ colors with respect to disks, or
with respect to scaled translations of a centrally
symmetric convex polygon. Moreover, Ajwani \etal \cite{Ajwani2012}
showed how to CF-color points with respect to \ifShort
\else rectangles;
the bound however goes up to~$O(n^{0.382})$.\fi

Recall that CF-colorings correspond to interference-free frequency
assignments in a cellular network. When a node in the network fails,
the resulting assignment may no longer be interference-free.
This leads to the study of \emph{$k$-fault-tolerant} CF-colorings,
where we want $\min(k,|\objects_q|)$ objects from $\objects_q$ to have a unique color.
In other words, a $k$-fault-tolerant CF-coloring allows the deletion
of $k$ objects without losing the conflict-free property.
Cheilaris \etal \cite{cheilaris-scf-14} studied the 1-dimensional case,
and presented a polynomial-time algorithm with approximation
ratio~$5-\frac{2}{k}$ for the problem of finding a CF-coloring with a minimum number of colors.
For~$k=1$---that is, the regular CF-coloring---the
algorithm gives a 2-approximation. Horev \etal \cite{hks-cfcms-10} studied the
2-dimensional case and proved a~$O(k \log n)$ bound for
disks and, more generally, regions with near-linear union complexity.

To increase coverage or capacity in a cellular network it may be necessary
to increase the number of base stations. This led Fiat \etal \cite{chen-ocf-07}
to study \emph{online} CF-colorings. Here the objects to be CF-colored
arrive over time, and as soon as an object appears it must receive a
color which cannot be changed later on. For CF-coloring points with respect to
intervals, they proposed a deterministic algorithm using~$O(\log^2 n)$
colors as well as two randomized algorithms, one of which is using at
most~$O(\log n \log \log n)$ colors in expectation and always
producing a valid coloring. Later, Chen~\etal~\cite{chen-ocf-06}
improved the bound with an algorithm using an expected~$O(\log n)$ colors.
\ifShort
\else
Chen~\etal~\cite{cks-ocfc-09} considered the 2-dimensional online problem in the dual
setting, i.e., coloring points with respect to geometric ranges.
\fi
They showed
that for ranges that are half-planes, unit disks, or bounded-size rectangles---i.e. rectangles
whose heights and widths all lie in the range $[1,c]$, for some fixed constant~$c$---\ifShort\else one
can obtain \fi an online CF-coloring using $O(\log n)$ colors with high probability.
For bounded-size rectangles they also presented a deterministic result
\ifShort \else for
online coloring,\fi using $O(\log^3 n)$ colors.
Bar-Noy~\etal \cite{barnoy-k-degeneracy} provided a general strategy for online CF-coloring
of hypergraphs. Their method uses~$O(k \log n)$ colors with high probability, where~$k$ is the so-called
\emph{degeneracy} of the hypergraph. Their method can for instance be applied for points with respect to half-planes using~$O(\log n)$ colors, which implies~\cite{cks-ocfc-09} the same result for unit disks with respect to points.
They also introduced a deterministic algorithm for points with respect to intervals in $\Reals^1$
if \emph{recolorings} are allowed. Their method
uses at most $n-\log n$ recolorings in total;  they did not obtain a bound
on the number of recolorings for an individual insertion.
Note that the results for online colorings in $\Reals^2$ are rather limited:
for the primal version of the problem---online CF-coloring objects with respect to points---there
are essentially only results for unit disks or unit squares (where the problem
is equivalent to the dual version of coloring points with respect to unit disks
and unit squares, respectively). Moreover, most of the results are randomized.
\medskip

De Berg~\etal~\cite{dynamic-cf-isaac} introduced the \emph{fully dynamic}
variant of the CF-coloring problem, which generalizes and extends the
fault-tolerant and online variants. Here the goal is to maintain a CF-coloring
under insertions and deletions. It is easy to
see that if we allow deletions and we do not
recolor objects, we may need to give each object in $\objects$ its own color. \ifShort \else
(Indeed, any two intersecting objects must have a different color
when all other objects are deleted.) \fi
Using $n$ colors is clearly undesirable. On the other hand,
recoloring all objects after each update---using then
the same number of colors as in the static case---is
not desirable either. Thus the main question is which trade-offs
can we get between the number of colors and the number of recolorings?
De~Berg~\etal proved a lower bound on this trade-off for the 1-dimensional
problem of CF-coloring intervals with respect to points.
(For this case it is straightforward to give a static CF-coloring with only three colors.)
Their lower bound implies that if we \ifShort
\else insist on using \fi $O(1/\eps)$ colors,
we must sometimes re-color $\Omega(\eps n^{\eps})$ intervals, and that if we allow
only $O(1)$ recolorings we must use $\Omega(\log n /\log\log n)$ colors
in the worst case. They also presented a strategy that uses
$O(\log n)$ colors at the cost of $O(\log n)$ recolorings.
The main goal of our paper is to study fully dynamic CF-colorings
for the 2-dimensional version of the problem.

\mypara{Our contributions.}
In Section~\ref{sec:rect} we give an algorithm for CF-coloring
unit squares using~$O(\log n)$ colors and~$O(\log n)$ recolorings
per update. Note that $\Omega(\log n)$ is a lower bound on the number
of colors for a CF-coloring of unit squares even in the static
case, so the number of colors our fully dynamic method uses is asymptotically optimal.
We also present an adaptation for bounded-size rectangles which
uses~$O(\log^2 n)$
\ifShort
\else colors, while still using~$O(\log n)$ recolorings per update.
\fi
The method also extends to arbitrary rectangles whose coordinates come
from a fixed universe of size $N$, yielding $O(\log^2 N \log^2 n)$
\ifShort
\else colors, still at the cost of $O(\log n)$ recolorings per insertion or deletion.
\fi
These constitute the first results on fully-dynamic CF-colorings in $\Reals^2$.

\ifFullyDynamic
    In Section~\ref{sec:general}, we give two general approaches
    that can be applied in many cases. The first uses a static coloring
    to solve insertions-only instances.
\fi
It can be applied in settings
where the static version of the problem admits a unimax coloring
with a small number of colors. The method can for example be used
to maintain a CF-coloring for pseudodisks with~$O(\log^3 n)$ colors and~$O(\log n)$
recolorings per update, or to maintain a CF-coloring for fat regions.
This is the first result for the semi-dynamic CF-coloring problem for such objects:
previous online results for coloring objects with respect to points in $\Reals^2$
only applied to unit disks or unit squares.
\ifFullyDynamic
    We extend the method to obtain a fully-dynamic solution,
    when the static solution allows what we call weak deletions.
    We can apply this technique for instance to CF-coloring
    points with respect to rectangles, using~$O(\sqrt{n} \log^2 n)$ colors
    and~$O(\log n)$ recolorings per insertion
    and~$O(1)$ recolorings per deletion.
\fi


\section{Dynamic CF-Colorings for  Unit Squares and Rectangles}\label{sec:rect}
In this section we explain how to color unit squares with~$O(\log n)$
colors and~$O(\log n)$ recolorings per update. We then generalise
this coloring to bounded-size rectangles, and to rectangles
with coordinates from a fixed universe.
We first explain our basic technique
on so-called anchored rectangles.

\subsection{A Subroutine: Maintaining a CF-coloring for Anchored Rectangles}\label{sec:quad}
We say that a rectangle $r$ is \emph{anchored} if its bottom-left vertex lies at the origin.
Let $\objects$ be a set of $n$ anchored rectangles.
We denote the $x$- and $y$-coordinate of the top-right vertex of a
rectangle~$r$ by $r_x$ and $r_y$, respectively. Our CF-coloring of $\objects$
is based on an augmented red-black tree, as explained next.
\medskip

To simplify the description we assume that the $x$-coordinates of the top-right vertices
(and, similarly, their $y$-coordinates) are all distinct---extending
the results to degenerate cases is straightforward.
We store $\objects$ in a red-black tree $\tree$ where $r_x$ (the $x$-coordinate
of the top-right vertex of~$r$) serves as the key of the rectangles~$r\in \objects$.
It is convenient to work with a \emph{leaf-oriented} red-black tree,
where the keys are stored in the leaves of the tree and the internal nodes store
splitting values.\footnote{Such a leaf-oriented red-black tree can be seen as a
regular red-black tree on a set $X’(\objects)$ that contains a splitting value between
any two consecutive keys. Hence, all the normal operations can be done in the standard way.}
We can assume without loss of generality that the splitting values lie strictly in between the keys.

For a node~$v\in\tree$,
let $\tree_v$ denote the subtree rooted at~$v$ and let $\objects(v)$ denote the set of
rectangles stored in the leaves of~$\tree_v$. We augment $\tree$ by storing a rectangle
$\rmax(v)$ at every (leaf or internal) node of $v$, define as follows:
\begin{equation*}
\mbox{$\rmax(v)$  :=  the rectangle $r\in \objects(v)$ that maximizes $r_y$.}
\end{equation*}

Let $\myleft(v)$ and $\myright(v)$ denote the left and right child, respectively, of an internal node~$v$.
Notice that $\rmax(v)$ is the rectangle whose top-right vertex has maximum $y$-value among $\rmax(\myleft(v))$ and $\rmax(\myright(v))$, so $\rmax(v)$ can be found in $O(1)$ time from the information at $v$'s children. Hence,
we can maintain the extra information in $O(\log n)$ time per insertion and deletion~\cite{clrs-ia-09}.

Next we define our coloring function. To this end we define for each rectangle $r\in\objects$ a set
$N(r)$ of nodes in $\tree$, as follows.
\[
N(r) := \{ v \in \tree : \mbox{ $v$ is the leaf storing $r$, or $v$ is an internal node with $\rmax(\myright(v)) = r$} \}.
\]
Observe that $N(r)$ only contains nodes on the search path to the leaf storing~$r$
and that $N(r)\cap N(r')=\emptyset$ for any two rectangles $r,r'\in \objects$.
Let $\height(v)$ denote the height of~$\tree_v$. Thus $\height(v)=0$ when $v$
is a leaf, and for non-leaf nodes~$v$ we have $\height(v) = \max(\height(\myleft(v)),\height(\myright(v))+1$.
We now define the color of a rectangle $r\in\objects$ as
\ifShort
\else
\[
\col(r) := \max_{v\in N(r)} \height(v).
\]
\fi
Since $N(r)$ always contains at least one node, namely the leaf storing~$r$, this
is a well-defined coloring.
\begin{lemma}\label{le:quadrant-CF-free}
The coloring defined above is conflict-free.
\end{lemma}
\ifShort
\else 
\begin{proof}
Recall that $\objects(v)$ denotes the set of rectangles stored in the subtree rooted at~$v$.
We prove by induction on $\height(v)$ that the coloring of $\objects(v)$ is conflict-free.
Since $\objects = \objects(\myroot(\tree(\objects)))$, this proves the lemma.
\medskip

When $\height(v)=0$ then $\objects(v)$ is a singleton, which is trivially colored conflict-free.
Now assume $\height(v)>0$. Let $x(v)$ denote the splitting value stored at~$v$,
and consider any point $q :=(q_x,q_y)$ in the plane. Let $\objects_q(v)\subseteq\objects(v)$
denote the set of rectangles from $\objects(v)$ containing~$q$, and assume $\objects_q(v) \neq\emptyset$.

If $q_x>x(v)$ then $q$ does not lie in any of the rectangles in $S(\myleft(v))$,
and so $\objects_q(v) = \objects_q(\myright(v))$. Since the coloring of $\objects(\myright(v))$
is conflict-free by induction, this implies that $\objects_q(v)$ has a rectangle with a unique color.

Now suppose $q_x \leq x(v)$. If $q_y > \rmax(\myright(v))_y$ then $q$ does not
lie in any rectangle from $S(\myright(v))$. Since the coloring of $\objects(\myleft(v))$
is conflict-free, this implies that $\objects_q(v)$ has a rectangle with a unique color.
Otherwise $q\in \rmax(\myright(v))$.
\begin{quotation}
\claim{ The rectangle $\rmax(\myright(v))$ has a unique color in $\objects(v)$ and, hence, in $\objects_q(v)$.}%
{ Let $u$ be the node that defines the color of~$\rmax(\myright(v))$, that is,
the node in $N(\rmax(\myright(v))$ with maximum height.
Since $v\in N(\rmax(\myright(v))$, either $u=v$ or $u$ is an ancestor of~$v$.
Let $r$ be any other rectangle in~$\objects(v)$ and
let $w$ be the node that defines the color of~$r$.
Because $r\in\objects(v)$, we know that $w$ is a node in $\tree_v$ or
an ancestor of~$v$. In the former case, since $N(r)$ is disjoint from
$N(\rmax(\myright(v))$ and hence does not contain~$v$, we conclude
that $\col(r)<\col(\rmax(\myright(v))$.
In the latter case we observe that $u$ and $w$ both are nodes
on the path from the root node to $v$, which means that $\height(u)\neq\height(w)$
and so $\col(r)\neq \col(\rmax(\myright(v)))$.
}
\end{quotation}
We conclude that the coloring is conflict-free.
\end{proof}
We obtain the following theorem.
\fi 
\begin{theorem}
Let $\objects$ be a set of anchored rectangles in the plane. Then it is
possible to maintain a CF-coloring on $\objects$ with $O(\log n)$ colors
using $O(\log n)$ recolorings per insertions and deletion,
where $n$ is the current number of rectangles in~$\objects$.
\end{theorem}
\begin{proof}
Consider the coloring method described
\ifShort
\else above.
Lemma~\ref{le:quadrant-CF-free} states that the coloring is conflict-free.
\fi
Since red-black trees have height $O(\log n)$, the number of colors used is
$O(\log n)$ as well.

Now consider an update on $\objects$. The augmented red-black tree can be updated
in $O(\log n)$ time in a standard manner~\cite{clrs-ia-09}. The color of a rectangle
$r\in \objects$ can only change when (i) the set $N(r)$ changes, or
(ii) the height of a node in $N(r)$ changes. We argue that this only happens
for $O(\log n)$ rectangles. Consider an insertion; the argument for deletions
is similar. In the first phase of the insertion algorithm for
red-black trees~\cite{clrs-ia-09} a new leaf is created
for the rectangle to be inserted. This may change $\height(v)$ or $\rmax(v)$
only for nodes $v$ on the path to this leaf, so
it affects the color of $O(\log n)$ rectangles. In the second phase the balance is restored using
$O(1)$ rotations. Each rotation changes $\height(v)$ or $\rmax(v)$ for only
$O(\log n)$ nodes, so also here only $O(\log n)$ rectangles are affected.
\end{proof}
\subsection{Maintaining a CF-Coloring for Unit Squares}
Let~$\objects$ be a set of unit squares. We first assume that all
squares in $\objects$ contain the origin.

A naive way to use the result from the previous section is to
partition each square $\sq\in\objects$ into four rectangular parts
by cutting it along the $x$-axis and the $y$-axis. Note that the
set of north-east rectangle parts (i.e., the parts to the north-east
of the origin) are all anchored
rectangles, so we can use the method described above to maintain a CF-coloring
on them. The other part types (south-east, south-west, and north-west) can be treated similarly.
Thus every square $\sq\in \objects$ receives four colors. If we now assign
a final color to $\sq$ that is the four-tuple consisting of those
four colors, then we obtain a CF-coloring with $O(\log^4 n)$ colors.
(This trick of using a ``product color'' was also used by, among others,
Ajwani \etal \cite{Ajwani2012}.)

It is possible to improve this by using the following fact: the ordering of
the~$x$-coordinates of the top-right corners of the squares in $\objects$ is the same
as the ordering of their bottom-right (or bottom-left, or top-left) corners.
This implies that instead of working with four different trees
we can use the same tree structure for all part types.
Moreover, even the extra information stored in
the internal nodes is the same for the north-east and north-west parts,
since the $y$-coordinates of the top-right and top-left vertices are the same.
Similarly, the extra information for the south-east and south-west parts
are the same. Therefore, we can modify the augmented
red-black tree to store two squares per internal node instead of
one:
\begin{itemize}
\item  $\mbox{$\smax(v)$  :=  the square $\sq\in \objects(v)$ that maximizes $\sq_y$}$,
\item  $\mbox{$\smin(v)$  :=  the square $\sq\in \objects(v)$ that minimizes $\sq_y$}$.
\end{itemize}
Next we modify our coloring function. Therefore we first
redefine the set~$N(s)$ of nodes for each square~$\sq\in \objects$:
\begin{align*}
N(s) := \{ \mbox{the leaf storing~$\sq$}\} \cup N_{\myne}(\sq) \cup N_{\myse}(\sq) \cup N_{\mysw}(\sq) \cup N_{\mynw}(\sq),
\end{align*}
where
\begin{itemize}
\item $N_{\myne}(\sq) := \{ v \in \tree :
    \mbox{$v$ is an internal node with $\smax(\myright(v)) = \sq$} \}$,
\item $N_{\myse}(\sq) := \{ v \in \tree :
    \mbox{$v$ is an internal node with $\smin(\myright(v)) = \sq$} \}$,
\item $N_{\mysw}(\sq) := \{ v \in \tree :
    \mbox{$v$ is an internal node with $\smin(\myleft(v)) = \sq$} \}$,
\item $N_{\mynw}(\sq) := \{ v \in \tree :
    \mbox{$v$ is an internal node with $\smax(\myleft(v)) = \sq$} \}$.
\end{itemize}

The coloring is as follows. We now
allow four colors per height-value, namely
for height-value~$h$ we give colors $4h+j$ for $j\in\{0,1,2,3\}$. These colors
essentially correspond to the colors we would give out for the four part types.
The color of a square~$\sq$ is now defined as
\[
\col(\sq) :=
\begin{cases}
0 & \mbox{if } \max\limits_{v\in N(s)} \height(v) = 0 \mbox{ ($\sq$ is only stored at a leaf),}
\\
4\cdot\max\limits_{v\in N(s)} \height(v) + j
& \mbox{if } \max\limits_{v\in N(s)} \height(v) > 0,
\end{cases}
\]
where
\[
j :=
\begin{cases}
0 & \mbox{if $\height(s) = \max\limits_{v\in N_{\myne}(s)} \height(v),$} \\
1 & \mbox{if the condition for $j=0$ does not apply and $\height(s) = \max\limits_{v\in N_{\myse}(s)} \height(v),$} \\
2 & \mbox{if the conditions for $j=0,1$ do not apply and $\height(s) = \max\limits_{v\in N_{\mysw}(s)} \height(v),$} \\
3 & \mbox{otherwise (we now must have $\height(s) = \max\limits_{v\in N_{\mynw}(s)} \height(v)$).}
\end{cases}
\]
The following lemma can be proven in exactly the same was ay Lemma~\ref{le:quadrant-CF-free}.
The only addition is that, when considering a set $\objects(v)$, we need to make a distinction
depending on in which quadrant the query point~$q$ lies. If it lies in the north-east quadrant
we can follow the proof verbatim, and the other cases are symmetric.
\begin{lemma}\label{le:u-squares-CF-free}
The coloring defined above is conflict-free.
\end{lemma}

It remains to remove the restriction that all squares contain the
origin. To this end we use a grid-based method, similar to the
one used by, e.g., Chen~\etal~\cite{cks-ocfc-09}. Consider the integer grid, and assign each
square in $\objects$ to the grid point it contains; if a square
contains multiple grid points, we assign it to the lexicographically smallest one.
Thus we create for each grid point $(i,j)$ a set $\objects(i,j)$ of
squares that all contain the point~$(i,j)$.
We maintain a CF-coloring for each such set using the method described above.
Note that a square in $\objects(i,j)$ can only intersect squares in
$\objects(i',j')$ when $(i',j')$ is one of the eight neighboring grid points of $(i,j)$.
Hence, when $i'=i\mod 2$ and $j'=j\mod 2$ we can re-use the same color set,
and so we only need four color sets of $O(\log n)$ colors each.
\begin{theorem}
Let $\objects$ be a set of unit squares in the plane. Then it is
possible to maintain a CF-coloring on $\objects$ with $O(\log n)$ colors
using $O(\log n)$ recolorings per insertions and deletion,
where $n$ is the current number of squares in~$\objects$.
\end{theorem}


\subsection{Maintaining a CF-Coloring for Bounded-Size Rectangles}
Let~$\objects$ be a set of \emph{bounded-size rectangles}:
rectangles whose widths and heights are between 1 and~$c$ for some fixed constant~$c$. Note that in practice, two different base stations
have roughly the same coverage, hence it makes sense to assume
the ratio is bounded by some constant~$c$.

First consider the case
where all rectangles in $\objects$ contain the origin.
Here, the $x$-ordering of the top-right corners of the rectangles
may be different from the $x$-ordering of the top-left corners
as we no longer use unit squares. Therefore the trees for the east
(that is, north-east and south-east) parts
no longer have the same structure. Note that the~$x$-ordering
of the top-left and bottom-left corners are the same, hence
only one tree suffices for the east parts, and the same holds
for the west parts.
Hence, we build and maintain two
separate trees, one for the east parts of the
rectangles and one for the west parts. In the east tree we only work
with the sets $N_{\myne}(\sq)$ and $N_{\myse}(\sq)$, and in the
west tree we only work with $N_{\mysw}(\sq)$ and $N_{\mynw}(\sq)$;
for the rest the structures and colorings are defined in the same as before.
We then use the product coloring to obtain our bound:
we give each rectangle a pair of colors---one coming from
the east tree, one coming from the west tree---resulting in $O(\log^2 n)$
different color pairs.

To remove the restriction that each rectangle contains the origin we
use the same grid-based approach as for unit squares. The only difference
is that a rectangle in a set $\objects(i,j)$ can now intersect rectangles from up to
$(1+2c)^2-1$ sets $\objects(i',j')$, namely with $i-c\leq i'\leq i+c$ and $j-c\leq j'\leq j+c$.
Since $c$ is a fixed constant, we still need only $O(1)$ color sets.
\begin{theorem}
Let $\objects$ be a set of bounded-size rectangles in the plane. Then it is
possible to maintain a CF-coloring on $\objects$ with $O(\log^2 n)$ colors
using $O(\log n)$ recolorings per insertion and deletion,
where $n$ is the current number of rectangles in~$\objects$.
\end{theorem}

\subsection{Maintaining a CF-Coloring for Rectangles with Coordinates from a Fixed Universe}
The solution can also be extended to rectangles of arbitrary sizes,
if their coordinates come from a fixed universe~$U := \{0,\ldots,N-1\}$
of size~$N$. Again, from a practical point of view it makes sense as
in a city for instance, the places a base station can be created
are limited.

To this end we construct a balanced tree $\tree_x$ over the universe~$U$,
and we associate each rectangle $r=[r_{x,1},r_{x,2}]\times [r_{y,1},r_{y,2}]$
to the highest node~$v$ in $\tree_x$ whose $x$-value $x(v)$ is contained in $[r_{x,1},r_{x,2}]$.
Let $\objects(v)$ be the set of objects associated to~$v$.
For each node $v\in\tree_x$ we construct a balanced tree $\tree_y(v)$ over the universe,
and we associate each rectangle $r\in\objects(v)$
to the highest node~$w$ in $\tree_y(v)$ whose $y$-value $y(w)$ is contained in $[r_{y,1},r_{y,2}]$.
(In other words, we are constructing a 2-level interval tree~\cite{bcko-cgaa-08}
on the rectangles, using the universe to provide the skeleton of the tree. The reason
for using a skeleton tree is that otherwise we have to maintain balance under
insertions and deletions, which is hard to do while ensuring worst-case bounds on the
number of recolorings.)
Let $\objects(w)$ be the set of objects associated to a node~$w$ in any
second-level tree~$\tree_y(v)$. All rectangles in $\objects(w)$ have a point
in common, namely the point $(x(v),y(w))$. Therefore we can proceed as in
the previous section, and maintain a CF-coloring on $\objects(w)$
with a color set of size $O(\log^2 n)$, using $O(\log n)$ recolorings per insertions and deletion.

Note that for any two nodes $w,w'$ at the same level in a tree $\tree_y(v)$,
any two rectangles $r\in\objects(w)$ and $r'\in\objects(w')$ are disjoint.
Hence, over all nodes $w\in\tree_x(v)$ we only need $O(\log N)$ different
color sets. Similarly, for any two nodes $v,v'$ of $\tree_x$ at the same level,
any two rectangles $r\in\objects(v)$ and $r'\in\objects(v')$ are disjoint.
Hence, the total number of color sets we need is $O(\log^2 N)$.
This leads to the following result.
\begin{theorem}
Let $\objects$ be a set of rectangles in the plane, whose coordinates
come from a fixed universe of size~$N$. Then it is
possible to maintain a CF-coloring on $\objects$ with $O(\log^2 N\log^2 n)$ colors
using $O(\log n)$ recolorings per insertions and deletion,
where $n$ is the current number of rectangles in~$\objects$.
\end{theorem}

\noindent \emph{Remark.}
Instead of assuming a skeleton tree and working with a fixed skeleton for
our 2-level  interval tree, we can also use randomized search trees.
Then, assuming the adversary doing the insertions and deletions is oblivious
of our structure and coloring, the tree is expected to be balanced at any
point in time. Hence, we obtain $O(\log^4 n)$ colors in expectation,
at the cost of $O(\log n)$ recolorings (worst-case) per update.



\section{A General Technique}\label{sec:general}
In this section we present a general technique to obtain a dynamic CF-coloring
scheme in cases where there exists a static unimax coloring.
(Recall that a unimax coloring is a CF-coloring where for any point $q$ the
object from $\objects_q$ with the maximum color is unique.)
Our technique results in a dynamic CF-coloring that uses $O(\cf(n) \log^2 n)$
colors, where $\cf(n)$ is the number of colors used in the static unimax coloring,
at the cost of $O(\log n)$ recolorings per update.
We first describe our technique for the case of insertions only. Then we extend
the technique to the fully-dynamic setting, for the case where the unimax coloring allows
for so-called weak deletions.

We remark that even though we describe our technique in the geometric setting in the plane,
the techniques provided in this section can be applied in the abstract hypergraph
setting as well.

\subsection{An Insertion-Only Solution}
Let~$\objects$ be a set of objects in the plane and assume
that~$\objects$ can be colored in a unimax fashion using~$\cf(n)$
colors, where~$\cf$ is a non-decreasing function.
Here it does not matter if $\objects$ is a set of geometric objects that we
want to CF-color with respect to points, or a set of points that we want to
CF-color with respect to a family of geometric ranges.
For concreteness we refer to the elements from $\objects$ as objects.

Our technique to maintain a CF-coloring under insertions
of objects into~$\objects$ is based on the \emph{logarithmic method}~\cite{bs-dspsd-80},
which is also used to make static data structures semi-dynamic.
Thus at any point in time we have $\lceil \log n\rceil+1$ sets $\objects_i$
such that each set $\objects_i$, for $i=0,\ldots,\lceil\log n\rceil$, is
either empty or contains exactly~$2^i$ objects. The idea is to give each set~$\objects_i$
its own color set, consisting of $\cf(2^i)$ colors.
Maintaining a CF-coloring
under insertions such that the \emph{amortized} number of recolorings is small,
is easy (and it does not require the coloring to be unimax):
when inserting a new object we find the first empty set $\objects_i$,
and we put all objects in $\objects_{0}\cup\cdots\cup\objects_{i-1}$
together with the new object into $\objects_i$. The challenge is to
achieve a \emph{worst-case} bound on the number of recolorings per insertion.
Note that for the maintenance of data structures, it is known how to achieve
worst-case bounds using the logarithmic method. The idea is to build
the new data structure for $\objects_i$ ``in the background'' and switch to the new structure
when it is ready. For us this does not work, however, since we would still need
many recolorings when we switch. Hence, we need a more careful approach.
\medskip

When moving all objects from $\objects_{0}\cup\cdots\cup\objects_{i-1}$
(together with the new object) into $\objects_i$, we do not recolor them all at once
but we do so over the next $2^i$ insertions. As long as we still need to recolor
objects from $\objects_i$, we say that $\objects_i$ is \emph{in migration}.
We need to take care that the coloring of a set that is in migration,
where some objects still have the color from the set $\objects_j$ they came from
and others have already received their new color in~$\objects_i$, is valid.
For this we need to recolor the objects in a specific order,
which requires the static coloring to be unimax as explained below.
Another complication is that, because the objects in $\objects_j$ that are being
moved to $\objects_i$ still have their own color, we have to be careful when we create
a new set $\objects_j$. To avoid any problems, we need several color sets
per set. Next we describe our scheme in detail.

As already mentioned, we have sets $\objects_0,\ldots,\objects_{\ell}$, where
$\ell:=\lceil \log n \rceil$. Each set can be in one of three states:
\emph{empty}, \emph{full}, or \emph{in migration}. For each $i$ with
$0\leq i\leq\ell$ we have $\ell-i+1$ color sets of size $\cf(2^i)$ available,
denoted by $\C(i,t)$ for $0\leq t\leq \ell-i$. The insertion of an object~$\ob$
into $\objects$ now proceeds as follows.
\begin{enumerate}
\item Let~$i$ be the smallest index such that~$S_i$ is empty. Note
      that~$i$ might be~$\ell +1$, in which case we introduce a new set
      and redefine~$\ell$. Note that this only happens when the number of
      objects reaches a power of 2.
\item Set $\objects_i := \{\ob\} \cup \objects_0\cup\cdots\cup\objects_{i-1}$.
      Mark $\objects_0,\ldots,\objects_{i-1}$ as \emph{empty},
      and mark~$\objects_i$ as \emph{in migration}.
\item \label{step:unimax} Take an unused color set $\C(i,t)$---we argue below that
      at least one color set $\C(i,t)$ with $0\leq t\leq \ell-i$ is currently unused---and compute a
      unimax coloring of~$\objects_i$ using colors from $\C(i,t)$. We refer to the color
      from $\C(i,t)$ that an object in $\objects_i$ receives as its \emph{final color}
      (for the current migration). Except for the newly inserted object~$\ob$,
      we do not recolor any objects to their
      final color in this step; they all keep their current colors.
\item For each set~$\objects_k$ in migration---this includes the set we just created
      in Step~\ref{step:unimax}---we recolor one object whose
      final color is different from its current color and whose
      final color is maximal among such objects. When multiple
      objects share that property, we arbitrarily choose one of them.
      If all objects in $\objects_k$ now have their final color, we mark $\objects_k$
      as \emph{full}.
\end{enumerate}
\begin{lemma}\label{lem:no_2_migrations}
Suppose that when we insert an object $\ob$ into $\objects$, the first
empty set is $\objects_i$. Then the sets $\objects_0,\ldots,\objects_{i-1}$ are full.
\end{lemma}
\begin{proof}
Suppose for a contradiction that $\objects_j$, for some $0\leq j<i$
is in migration. Consider the last time at which $\objects_j$ was
created---that is, the last time at which we inserted an object $\ob'$ that caused
the then-empty set $\objects_j$ to be created and marked as in migration.
Upon insertion of~$\ob'$, we already perform one recoloring in~$\objects_j$.
At that point all sets~$S_0,\ldots, S_{j-1}$ were marked empty and it
takes~$\sum_{t=0}^{j-1} 2^t=2^j -1$ additional insertions to fill them,
giving us as many recolorings in~$\objects_j$. Thus before we create any
set $\objects_i$ with $i>j$, we have already marked $\objects_j$ as full.
Since $\ob'$ was the last object whose insertion created $\objects_j$,
by the time we create $\objects_i$ the set $\objects_j$ must still
be full---it cannot in the mean time have become empty and later be re-created (and thus be in migration).
\end{proof}
Next we show that in Step~\ref{step:unimax} we always have an unused
color set at our disposal.
\begin{lemma}\label{lem:forr-color-set}
When we create a new set $\objects_i$ in Step~\ref{step:unimax},
at least one of the color sets $\C(i,t)$ with $0\leq t\leq \ell-i$ is currently unused.
\end{lemma}
\begin{proof}
Consider a color set $\C(i,t)$. The reason we may not be able to use $\C(i,t)$ when we create
$\objects_i$ is that there is a set $\objects_{i'}$ with $i'>i$ that is currently in migration:
the objects from a previous instance of $\objects_i$ (that were put into $\objects_{i'}$
when we created $\objects_{i'}$) may not all have been recolored yet.
By Lemma~\ref{lem:no_2_migrations} this previous instance was full when it
was put into~$\objects_{i'}$ and so it only blocks a single color set,
namely one for $\objects_i$.
Thus the number of color sets $\C(i,t)$ currently in use is at most $\ell-i$.
Since we have $\ell-i+1$ such colors sets at our disposal, one must be unused.
\end{proof}
\ifShort \else
We obtain the following result.
\fi
\begin{theorem}\label{thm:insertion-only}
Let $\family$ be a family of objects such that any subset of $n$
objects from~$\family$ admits a unimax coloring with $\cf(n)$ colors,
where $\cf(n)$ is non-decreasing.
Then we can maintain a CF-coloring on a set $\objects$  of objects from
$\family$ under insertions, such that the number of used colors
is $O(\cf(n)\log^2 n)$ and the number of recolorings per insertion is at most $\lceil\log n\rceil$,
where~$n$ is the current number of objects in~$\objects$.
\end{theorem}
\begin{proof}
The number of colors used is $\sum_{i=0}^{\ell} (\ell-i+1)\cf(2^i)$,
where $\ell = \lceil\log n\rceil$. Since $\cf(n)$ is non-decreasing,
this is bounded by $O(\cf(n)\log^2 n)$. The number of recolorings
per insertion is at most one per set $\objects_i$, so at most $\lceil\log n\rceil$ in total.
(The total number of sets is actually $\lceil\log n\rceil+1$, but not all
of them can be in migration.)
\medskip

It remains to prove that the coloring is conflict-free.
Consider a point~$q\in\Reals^2$. (Here we use terminology
from CF-coloring of objects with respect to points. In the dual
version, $q$ would be a range.) Let $\objects_i$ be a set containing
an object $\ob$ with $q\in\ob$; if no such set exists there is
nothing to prove.

If~$\objects_i$ is full then
it has a unimax coloring using a color set~$\C(i,t)$
not used by any other set $\objects_j$. Hence, there is
an object containing~$q$ with a unique color.

Now suppose that $\objects_i$ is in migration. We have two cases:
(i)~$q$ is contained in an object from $\objects_i$ that has already received its final color,
(ii)~all objects in $\objects_i$ containing~$q$ still have their old color.

In case~(i) the object containing~$q$ with the highest final color must have
a unique color, because of the following easy-to-prove fact.

\fact{Consider any set $A$ colored with a unimax coloring. Let $z$ be an integer,
and let $B\subseteq A$ be a subset that contains all objects of color
greater than~$z$, some objects of color~$z$, and at most one other object.
Then the coloring of $B$ is unimax.}
This fact proves the statement above for case~(i), because
we recolor the objects in decreasing order of their colors
and the coloring we are migrating to is a unimax coloring.
(The ``at most one other object'' mentioned in the fact is needed because
the object that caused the migration immediately receives its color,
and this color needs not be the highest color.)

In case~(ii), $q$ is contained in an object from some old set $\objects_j$
with $j<i$. At the time we created $\objects_i$ this set $\objects_j$
was CF-colored, and since we did not yet recolor any object from $\objects_j$
that contains~$q$ ---otherwise we are in case~(i)---we conclude that $q$
is contained in an object with a unique color.
\end{proof}

\mypara{Application: Objects with Near-Linear Union Complexity.}
Har-Peled and Smorodinsky~\cite{harpeled-cf-05} proved that
any family of objects with linear union complexity
(for example disks, or pseudodisks)
can be colored in a unimax fashion using~$O(\log n)$
colors. In fact, their result is more general: if the union complexity
is  at most $n\cdot\beta(n)$ then the number of colors is $O(\beta(n)\log n)$.
Note that for disks and pseudodisks we have $\beta(n)=O(1)$,
for fat triangles we have $\beta(n)=O(\log^* n)$~\cite{abes-ubulfo-2014}
and for locally fat objects we have $\beta(n)=O(2^{O(\log^* n)})$~\cite{abes-ubulfo-2014}.
This directly implies the following result.

\begin{corollary}
Let $\family$ be a family of objects such that the union complexity of any subset of $n$
objects from~$\family$ is at most $n\beta(n)$. Then we can maintain a CF-coloring on a set $\objects$
of objects from~$\family$ under insertions, such that the number of used colors
is $O(\beta(n)\log^3 n)$ and the number of recolorings per insertion is $O(\log n)$,
where~$n$ is the current number of objects in~$\objects$.
\end{corollary}


\subsection{A Fully-Dynamic Solution}

\ifShort
\else 
We now generalize the semi-dynamic solution presented above so that it
can also handle deletions. As before we assume we have a family~$\family$
of objects such that any set of $n$ objects from $\family$ can
be unimax-colored with $\cf(n)$ colors. We further assume that such a coloring
admits \emph{weak deletions}: once we have colored a given set $\objects$ of $n_0$
objects using $\cf(n_0)$ colors, we can delete objects from it using $r(n_0)$
recolorings per deletion such that the number of colors never exceeds~$\cf(n_0)$.
The functions $\cf(n)$ and $r(n)$ are assumed to be non-decreasing.
\medskip

Let $\objects$ be the current set of objects. We again employ
ideas from the logarithmic method, but we need to relax the conditions
on the set sizes. More precisely, we maintain an
integer~$\ell$ and
a partition of $\objects$ into $\ell+1$ sets $\objects_0,\ldots,\objects_{\ell}$,
such that the following \emph{size invariant} is maintained.
\begin{quotation}
\noindent
\vspace*{-5mm}
\begin{description}
\item[(Inv-S)] For all $0\leq i< \ell$ we have $|\objects_i|\leq 2^i$,
                and we have $2^{\ell-2} \leq |\objects_{\ell}|\leq 2^\ell$.
\end{description}
\end{quotation}
Note that the second part of (Inv-S) implies that $\ell=\Theta(\log n)$, where $n:= |\objects|$.
%
As before, for each $i$ with $0\leq i\leq\ell$ we have color sets $\C(i,t)$ available,
each of size $\cf(2^i)$. This time the number of colors sets $\C(i,t)$ is $\ell+2$ for each~$i$, instead
of $\ell-i$. Additionally, we allow the use of colors sets~$\C(\ell+1,t)$.
Hence, for each~$i=0,\ldots, \ell+1$, we have~$\ell+2$ color
sets of size~$\cf(2^i)$.
A set $\objects_i$ can now be in four states:
\emph{empty}, \emph{non-empty}, \emph{in upwards migration},
and \emph{in downwards migration}. It is worth pointing out
that only~$\objects_\ell$ can be in downwards migration.

Our coloring of the sets~$\objects_i$ satisfies
several \emph{color invariants}, which depend on the state of~$\objects_i$.
The invariant for sets $\objects_i$ whose state is \emph{non-empty}
is relatively straightforward.
\begin{quotation}
\noindent
\vspace*{-5mm}
\begin{description}
\item[(Inv-C-NonEmp)] The objects in a set $\objects_i$ whose state is \emph{non-empty}
              are unimax-colored using a color set $\C(i,t)$ not used elsewhere.
\end{description}
\end{quotation}
Before we can state the invariants for sets $\objects_i$ in migration
we need to introduce some notation.

A set $\objects_i$ that is in upwards migration is the disjoint union
of subsets $\objects_i^{(0)},\ldots,\objects_i^{(\ell)}$, some of which may be empty,
plus at most one other object. When set~$\objects_\ell$ is in downwards
migration, it is the disjoint union of either
subsets~$S_\ell^{(\ell-2)}, S_\ell^{(\ell-1)}$, and~$S_\ell^{(\ell)}$,
or of subsets~$S_\ell^{(\ell-1)}, S_\ell^{(\ell)}$, and~$S_\ell^{(\ell+1)}$.
In the former case, we define~$\ell':=\ell$, in the latter case~$\ell'=\ell+1$.
Defining $\ell'$ in this way simplifies the descriptions.

The idea is now that we have a ``global'' coloring for the
set $\objects_i$---this is the new coloring we
are migrating to---and separate ``local'' colorings for each of the sets~$\objects_i^{(m)}$.
The global color of an object is called its \emph{final color}, and
the local color is called its \emph{temporary color}. To know which of these two colors
is the actual color of an object, we maintain a (possibly empty) subset~$\objects_i^*\subseteq \objects_i$ such that
the actual color of an object in $\objects_i^*$ is its final color, and
the actual color of an object in $\objects_i\setminus \objects_i^*$ is its temporary color.
Thus, when $\objects_i^*=\objects_i$ then all objects
received their new color and the migration is finished.

We first give the color invariants that hold for all sets in migration,
and then give additional invariants that depend on whether the set is in upwards
or downwards migration.
\begin{quotation}
\noindent
\vspace*{-5mm}
\begin{description}
\item[(Inv-C-Mig-1)]
    If $\objects_i$ is in migration, then we have a unimax coloring on $\objects_i$ using a color set
    $\C(i,t)$ not used elsewhere. As already mentioned, the color an object receives
    in this coloring is called its \emph{final color}.
\item[(Inv-C-Mig-2)]
    If $\objects_i$ is in migration, then there is an integer $z$ such that $\objects_i^*$ contains all objects from~$\objects_i$ whose final color
    is greater than~$z$, some objects of color~$z$, and at most one other object.
\end{description}
\end{quotation}
For sets in upwards migration we also have the following invariant.
\begin{quotation}
\noindent
\vspace*{-5mm}
\begin{description}
\item[(Inv-C-Up)]
    If $\objects_i$ is in upwards migration, then for each $\objects_i^{(m)}$
	we have a unimax coloring using a color set $\C(m,t)$ not used elsewhere.
    The color an object receives in this coloring is its \emph{temporary color}.
\end{description}
\end{quotation}
Finally, when $\objects_{\ell}$ is in downwards migration---recall that $\objects_i$
can only be in downwards migration when $i=\ell$---we have the following invariant.
\begin{quotation}
\noindent
\vspace*{-5mm}
\begin{description}
\item[(Inv-C-Down)]
    If $\objects_\ell$ is in downwards migration, then the following
    holds.
    \begin{itemize}
    \item  Each set~$\objects_\ell^{(m)}$ with $m<\ell'-2$ is empty.
    \item  The set~$\objects_\ell^{(\ell'-2)}$ is unimax colored
    	   using at most $\ell'-1$ color sets not used elsewhere. More precisely, for each
           $m=0,\ldots, \ell'-2$ at most one color set $\C(m,t)$ is used
           in the coloring of $\objects_\ell^{(\ell'-2)}$.
    \item  The set~$\objects_\ell^{(\ell'-1)}$ is unimax colored
    	   using at most $\ell'$ color sets not used elsewhere. More precisely, for each
           $m=0,\ldots, \ell'-1$ at most one color set $\C(m,t)$ is used
           in the coloring of $\objects_\ell^{(\ell'-1)}$.
    \item  The set~$\objects_\ell^{(\ell')}$ is unimax colored
    	   using at most one color set $\C(\ell',t)$     	
           not used elsewhere.
    \end{itemize}
    The color an object receives in these colorings is its \emph{temporary color}.
\end{description}
\end{quotation}
Note that, when $\objects_{\ell}$ is in downwards migration,
the sets $\objects_{\ell}^{(\ell'-2)}$ and $\objects_{\ell}^{(\ell'-2)}$
are colored using several color sets.
In fact, it is the case that the algorithm distinguishes several subsubsets
of $\objects_{\ell}^{(\ell'-2)}$ (and, similarly, of $\objects_{\ell}^{(\ell'-2)}$),
each of which is unimax-colored using a different color set.
It is easily checked that the coloring obtained in this manner for $\objects_{\ell}^{(\ell'-2)}$
is a unimax-coloring, which allows weak deletions by doing a weak deletion
on the relevant subsubset.
\medskip

The next lemma implies that it is sufficient to maintain a coloring
satisfying all invariants; its proof
is similar to the proof of Theorem~\ref{thm:insertion-only}
that the coloring used by our insertion-only method is conflict-free,
with sets marked non-empty taking the role of sets marked full.
\begin{lemma}\label{le:fully-dynamic-conflict-free}
A coloring satisfying all color invariants is conflict-free.
\end{lemma}
Next we describe how to insert or delete an object~$\ob$.
We start with insertions.
\begin{enumerate}
\item \label{step:ins1} Let~$i$ be the smallest index such that~$\objects_i$ is empty.
      Find the smallest~$j$ such that~$\objects_1 \cup \cdots \cup \objects_{i-1}$ fit into~$\objects_j$.
      Note that $j\leq i$. If $j=\ell+1$, set $\ell:=\ell+1$.
\item \label{step:ins2} Set $\objects^*_j := \{\ob\}$, set $\objects_j^{(m)} := \objects_{m}$ for
      all~$0\leq m \leq i-1$, and set
      $\objects_j := \objects^*_j \cup \objects_j^{(0)}\cup\cdots\cup\objects_j^{(i-1)}$.
      Mark $\objects_0,\ldots,\objects_{i-1}$ as \emph{empty}, and then
      mark~$\objects_j$ as \emph{in upwards migration}.
      Mark all colors of the objects in $\objects_j^{(0)}\cup\cdots\cup\objects_j^{(i-1)}$ as temporary.
\item \label{step:unimax-2} Take an unused color set $\C(j,t)$ and compute a
      unimax coloring of~$\objects_j$ using colors from $\C(j,t)$. The color
      that an object in $\objects_j$ receives in this coloring is its \emph{final color}.
      Except for the newly inserted objects~$\ob$,
      we do not recolor any objects to their
      final color in this step; they all keep their temporary colors.
\item \label{step:ins4} For each set~$\objects_k$ in migration---this includes
      the set we just created in Step~\ref{step:unimax}---we proceed as follows.
      \begin{itemize}
      		\item If~$k<\ell$, we pick one object whose
      			  final color is different from its actual
      			  color and that has the highest
	              final colors among such objects. We recolor the object
                  so that its actual color becomes its final color,
                  and we add it to~$\objects_k^{*}$.
	        \item If~$k=\ell$, we pick two objects whose
      			  final color is different from its actual
      			  color and that have the highest
	              final colors among such objects. We recolor
                  them to their final color, and add them to~$\objects_k^{*}$.
      \end{itemize}
      In both cases we make an arbitrary choice in case of ties, and in the second case
      when only one object still needs to be recolored we just recolor that one.
      If all objects in $\objects_k$ now have their final color---thus
      $\objects_k^{*} =\objects_k$---we mark $\objects_k$ as (that is, we change its status to)
      \emph{non-empty}.
\end{enumerate}
We need the following lemma.
\begin{lemma}\label{lem:no_2_migrations-full}
Suppose that when we insert an object $\ob$ into $\objects$, the first
empty set is $\objects_i$. Then the sets $\objects_0,\ldots,\objects_{i-1}$ are marked \emph{non-empty}.
\end{lemma}
\begin{proof}
Suppose for a contradiction that $\objects_k$,
for some $0\leq k<i$ is in migration.
If~$\objects_k$ is in upwards migration
then the proof is similar to the proof
of Lemma~\ref{lem:no_2_migrations}. Suppose
now~$\objects_k$ is in downwards migration (in which case,~$k=\ell$).
At the moment~$\objects_\ell$ is marked as being in downwards migration,
we know~$\objects_{\ell-1}$ is empty---see the description
of a deletion further down. Hence, at least~$2^{\ell-1}$
objects need to be inserted before an insertion can reach~$\objects_{\ell}$.
Since we do two recolorings per insertion, this implies that~$S_\ell$
is no longer in migration when~$\ob$ is inserted, yielding a contradiction.
\end{proof}
We also need the analog of Lemma~\ref{lem:forr-color-set}.
\begin{lemma}
In Step~\ref{step:unimax} of the insertion procedure,
at least one of the color sets $\C(j,t)$ with $0\leq t\leq \ell+1$ is currently unused.
\end{lemma}
\begin{proof}
A color set $\C(j,t)$ can only be already in use
for sets $\objects_m$ with $m\neq j$ that are in migration. By Lemma~\ref{lem:no_2_migrations-full}
we have $m>j$ for such sets. When $m<\ell$ then $\objects_m$ uses
at most one color set~$\C(j,t)$, namely for $\objects_m^{(j)}$;
see Invariant~(Inv-C-Up). Hence, there are
at most $\ell-2$ color sets~$\C(j,t)$ already in use by sets~$\objects_{m}$
with $m\neq \ell$.
It remains to consider~$\objects_{\ell}$. In fact, it suffice
to consider the case where~$\objects_\ell$ is in downwards migration
as the other cases only use at most one color set~$\C(j,t)$ for~$\objects_{\ell}$.
If~$\objects_{\ell}$ is in downwards migration,
then only~$\objects_\ell^{(\ell'-2)}$ and~$\objects_\ell^{(\ell'-1)}$
can use a color set~$\C(j,t)$; see (Inv-C-Down). Then, the total number
of color sets~$\C(j,t)$ used is at most~$\ell-2+2=\ell$, leaving at least
one unused color set.
\end{proof}
We can now prove the correctness of the insertion procedure.
\begin{lemma}\label{le:insertion-inv}
The insertion procedure maintains all size and color invariants.
\end{lemma}
\begin{proof}
It is easy to check that the size invariant is maintained.
Color invariant~(Inv-C-NonEmp) remains true because we only create sets marked
\emph{non-empty} in Step~\ref{step:ins4} and when we do they are unimax-colored.
Invariant~(Inv-C-Mig-1) still holds as well, because we only create
a new set in migration in Step~\ref{step:ins2} and then
in Step~\ref{step:unimax-2} we generate a unimax coloring for it using a set
not used elsewhere. Invariant~(Inv-C-Mig-2) holds because of the
way Step~\ref{step:ins4} works.
We only need to check Invariant~(Inv-C-Up) for the set $\objects_j$
that is marked as being in upwards migration in Step~\ref{step:ins2},
and there it holds because (Inv-C-NonEmp) holds before
the insertion. Finally, (Inv-C-Down) cannot be violated because
it holds before the insertion and
our insertion algorithm does not mark a set as being in downwards migration.
\end{proof}
Next we describe the deletion of an object~$\ob$.
\begin{enumerate}
\item Let~$i$ be such that~$\ob\in \objects_i$.
\item \begin{enumerate}
      \item\label{step:weak-del} If~$i\neq \ell$ or ($i=\ell$ and~$|\objects_\ell| > 2^{\ell-2}$),
             and in addition~$\objects_i$ is not
            in migration, then do a weak deletion of~$\ob$ in~$\objects_i$.
            Mark $\objects_i$ as empty if applicable.
      \item \label{step:2b}
            If~$i\neq \ell$ or ($i=\ell$ and~$|\objects_\ell| > 2^{\ell-2}$), and in addition~$\objects_i$ is in migration,
            then do the following.
            \begin{itemize}
            \item Do a weak deletion of~$\ob$ in~$\objects_i^{(m)}$,
                  where~$\objects_i^{(m)}$ is the subset of~$\objects_i$ containing~$\ob$.
                  Note that this involves changing the actual color of
                  at most~$r(2^m)$ of the objects
                  in~$\objects_i^{(m)} \setminus \objects_i^*$.
            \item Do a weak deletion of~$\ob$ on $\objects_i$, thus changing the final color
                  of at most $r(2^i)$ objects. Note that this may break invariant~(Inv-C-Mig-2).
                  Repair (Inv-C-Mig-2) by removing at most $r(2^i)$ objects from $\objects_i^*$ and putting at most $r(2^i)$ other objects
                  from $\objects_i$ into $\objects_i^*$ instead.
                  We do this such that the size of $\objects_i^*$ does not change.
                  Observe that objects added to $\objects_i^*$  are recolored to their final color, while objects removed from $\objects_i^*$  are recolored to
                  their temporary color.
            \end{itemize}
             Mark~$\objects_i$ as empty if applicable.
      \item \label{step:del2c} Otherwise we have $i=\ell$ and $|\objects_{\ell}|=2^{\ell-2}$, so the deletion of~$\ob$
            breaks the condition on the size of $\objects_\ell$.
            In this case we merge the last three
            sets~$\objects_{\ell-2}, \objects_{\ell-1},\objects_\ell$, as follows.
            Set $\ell' := \ell$.
            If the three sets~$\objects_{\ell-2}, \objects_{\ell-1},\objects_\ell$
            together fit into $\objects_{\ell-1}$ then set $\ell := \ell-1$, otherwise
            keep $\ell$ as it is.
            \begin{itemize}
            \item If~$\objects_{\ell'-2}$ is \emph{empty} or \emph{non-empty},
              set $\objects_{\ell}^{(\ell'-2)}:=\objects_{\ell'-2}$; otherwise,
              set $\objects_{\ell}^{(\ell'-2)}:= \cup_{i=0}^{\ell'-2} \objects_{\ell'-2}^{(i)}$.
              Note that in the later case,~$\objects_{\ell}^{(\ell'-2)}$
              is now using at most~$\ell'-2$ color sets, namely
              at most one color set $\C(m,t)$ for each $m \leq \ell'-2$.
            \item If~$\objects_{\ell'-1}$ is \emph{empty} or \emph{non-empty},
              set $\objects_{\ell}^{(\ell'-1)}:=\objects_{\ell'-1}$; otherwise,
              set $\objects_{\ell}^{(\ell'-1)}:= \cup_{i=0}^{\ell'-1} \objects_{\ell'-1}^{(i)}$.
              Note that in the later case,~$\objects_{\ell}^{(\ell'-1)}$
              is now using at most~$\ell'-1$ color sets, namely
              at most one color set $\C(m,t)$ for each $m\leq \ell'-2$.
            \item Set~$\objects_{\ell}^{(\ell')}:=\objects_{\ell'}$
            and~$\objects_{\ell}^*:=\emptyset$.
            \end{itemize}
            Now compute the final coloring on $\objects_\ell$
            using an unused color set~$\C(\ell,t)$ and
            mark~$\objects_\ell$ as being in downwards migration.
       \end{enumerate}
\item\label{step:del-recol} If~$i= \ell$, do two recolorings in~$\objects_\ell$
starting from the highest final colors and add the two recolored
objects to~$\objects_\ell^*$. If only one such object remains,
do the recoloring only on that one.
If~$\objects_\ell^*=\objects_\ell$, mark~$\objects_\ell$
non-empty and free the unused color sets~$\C(\ell,t)$.
\end{enumerate}

\begin{lemma} \label{le:delete}
Suppose that upon deletion of~$\ob$ from~$\objects_\ell$, Step~\ref{step:del2c}
applies and thus we mark~$\objects_\ell$ as in downwards migration.
Then~$\objects_\ell$ was not in migration just before~$\ob$ is deleted.
\end{lemma}
\begin{proof}
Suppose for a contradiction that~$\objects_\ell$ is in migration.
Let~$\ob'$ be the last object that caused a migration on~$\objects_\ell$.
If $\ob'$ was being inserted then $|\objects_\ell|\geqslant 2^{\ell-1}$, because
of the choice of~$j$ in Step~\ref{step:ins1} of the insertion algorithm.
If $\ob'$ was being deleted then we also have $|\objects_\ell|\geqslant 2^{\ell-1}$,
because of (Inv-S) and because we decrement $\ell$
in Step~\ref{step:del2c} of the deletion procedure when ~$\objects_{\ell-2}, \objects_{\ell-1},\objects_\ell$ fit into $\objects_{\ell-1}$.
Hence, before~$\ob$ is deleted, at least~$2^{\ell-2}-1$ other objects
have been deleted, generating $2^{\ell}$ recolorings
for~$\objects_\ell$. This is enough to recolor
all objects in~$\objects_\ell$, contradicting that $\objects_\ell$ is still in migration.
\end{proof}
\begin{lemma}\label{le:deletion-col}
When we compute the final coloring of~$\objects_{\ell}$
in Step~\ref{step:del2c} of the deletion procedure,
at least one of the color sets $\C(\ell,t)$ with $0\leq t\leq \ell+1$
is currently unused.
\end{lemma}
\begin{proof}
Before the deletion that caused the move, thanks to Lemma~\ref{le:delete}
and the fact that only the last set can be in downwards migration,
no set is in downwards migration. Therefore,
before the deletion, each set~$\objects_0,\ldots,\objects_\ell$
uses at most one color set~$\C(\ell,t)$. Hence at
most~$\ell+1$ colors sets are being used, leaving
at least one color set available.
\end{proof}
\begin{lemma}\label{le:deletion-inv}
The deletion procedure maintains all size and color invariants.
\end{lemma}
\begin{proof}
Invariant~(Inv-S) is maintained since when the
set~$\objects_\ell$ becomes too small, Step~\ref{step:del2c} redistributes
the last three sets such that the invariant holds again.
Invariant~(Inv-C-NonEmp) is maintained by construction in
Step~\ref{step:weak-del}. Invariants~(Inv-C-Mig-1), (Inv-C-Mig-2)
and (Inv-C-Up) are maintained
by construction in Step~\ref{step:2b}.
We now argue that (Inv-C-Down) is maintained. It is obvious that
the first part of (Inv-C-Down) is maintained. The second and third part are
maintained because~(Inv-C-NonEmp) and~(Inv-C-up) hold before the deletion;
(Inv-C-NonEmp) is needed when $\objects_{\ell'-2}$ resp.~$\objects_{\ell'-1}$
are \emph{empty} or \emph{non-empty}, otherwise we need~(Inv-C-up).
The last part is maintained because before the deletion,
due to Lemma~\ref{le:delete} and the fact Invariant~(Inv-C-NonEmp)
holds before the deletion.
\end{proof}

We obtain the following result.
\fi  
\begin{theorem}\label{thm:g_i+d}
Let $\family$ be a family of objects such that any subset of $n$
objects from~$\family$ admits a unimax coloring with $\cf(n)$ colors
and that allows weak deletions at the cost of $r(n)$ recolorings,
where $\cf(n)$ and $r(n)$ are non-decreasing.
Then we can maintain a CF-coloring on a set $\objects$  of objects from
$\family$ under insertions and deletions, such that the number of used colors
is~$O(\sum_{i=0}^{k}\cf(2^i)\log n)$, where $k=\Theta(\log n)$.
The number of recolorings per insertion is $O(\log n)$,
and the number of recolorings per deletion is $O(r(8n)+1)$,
where~$n$ is the current number of objects in~$\objects$.
\end{theorem}
\ifShort \else
\begin{proof}
The correctness of our insertion and deletion procedures follows from
Lemma~\ref{le:fully-dynamic-conflict-free} together with Lemmas~\ref{le:insertion-inv}
and~\ref{le:deletion-inv}.
The number of colors used is at most~$\sum_{i=0}^{\ell+1} (\ell+2)\cf(2^i)$,
where $\ell = \Theta(\log n)$.
Since $\cf(n)$ is non-decreasing,
this is bounded by~$O(\sum_{i=0}^{\ell}\cf(2^i)\log n)$.
The number of recolorings
per insertion is at most~$2$ per set $\objects_i$,
so $O(\log n)$ in total and at most~$2r(2^{\ell'})$ per deletion for
at most three sets, so $O(r(8n))$ in total.
\end{proof} \fi  

\mypara{Application: Points with Respect to Rectangles.}
We now make use of Theorem~\ref{thm:g_i+d} to maintain a CF-coloring
of points with respect to rectangles. But first we present a
simple technique to color points with respect to intervals in ~$\Reals^1$,
which we use as a subroutine.

\begin{lemma}\label{lemma:pwrtint}
We can maintain a unimax coloring of~$n$ points in $\Reals^1$ with respect
to intervals under deletions, using $\lceil \log n_0 \rceil$ colors and
at the cost of one recoloring per deletion. Here $n_0$ is the initial
number of points.
\end{lemma}
\begin{proof}
We start with a static unimax coloring of points with respect
to intervals using~$\lceil \log n_0 \rceil$ colors~\cite{S-survey-10}.
Recoloring after deleting a point~$p$ with color~$i$ is done as follows.
If both neighbors of $p$ have a higher color then we do nothing, otherwise
we pick a neighbor with color smaller than~$i$ and recolor it to~$i$.
To prove the coloring stays unimax we only need to consider intervals~$I$ containing a 
neighbor of~$p$; for other intervals nothing changed. Now consider~$I\cup\{p\}$.
If before the deletion the maximum color was larger than~$i$ then
that color is still present and unique. Otherwise~$i$ was the unique
maximum color. Now either $I$ contains a neighbor of $p$ that was recolored
to~$i$, or no point in $I$ was recolored; in both cases the maximum color in~$I$
is unique.
\end{proof}

\ifShort
\else 
\emph{Remark.} that we can also get a fully dynamic solution using~$O(\log n)$ 
colors and~$O(\log n)$ recolorings per insertion or deletion, by storing
the points in a red-black tree and coloring them with their height in the tree.
\medskip

\fi
We now explain how to color points in the plane with respect to
rectangles. Let~$\objects$ be a set of points and~$\family$
be the family of all rectangles in the plane. The following
lemma shows how to perform weak deletions.

\begin{lemma}
There is a conflict free coloring of~$n$ points with respect
to rectangles using~$O(\sqrt{n}\log n)$ colors that allows weak
deletions at the cost of one recoloring per deletion.
\end{lemma}
\begin{proof}
We first partition the point set into at most~$\sqrt{n}$ subsets
such that each set is monotone using Dilworth's theorem \cite{dilworth-thm}.
Then, each point set behaves exactly as points with respect to
intervals in one dimension. Indeed, if a rectangle contains two
points, since the sequence of points is monotone, it also
contains all the points in between. We can then apply
Lemma~\ref{lemma:pwrtint} to finish the proof.
\end{proof}

\ifShort \else
We can directly conclude the following corollary. \fi

\begin{corollary}
Let~$\objects$ be a set of points in the plane and
$\family$ be a family of rectangles.
Then we can maintain a CF-coloring on~$\objects$
under insertions and deletions such that the number of used colors
is $O(\sqrt{n}\log^2 n)$ and the number of recolorings per
insertion is $O(\log n)$ and~$O(1)$ per deletion,
where~$n$ is the current number of objects in~$\objects$.
\end{corollary}



\section{Concluding Remarks}
We studied the maintenance of a CF-coloring under insertions and \ifShort
\else deletions of objects,\fi presenting the
first fully-dynamic solution for objects in $\Reals^2$.
We showed how to maintain a CF-coloring for unit
squares and for bounded-size rectangles, with $O(\log n)$ resp.~$O(\log^2 n)$
colors and $O(\log n)$ recolorings per update.
The method extends to arbitrary rectangles with coordinates 
from a fixed universe of size $N$, yielding $O(\log^2 N \log^2 n)$
colors and $O(\log n)$ recolorings per update.
We also presented general techniques for the semi-dynamic 
(insertions-only) and the fully-dynamic case (insertions and deletions).
Our insertions-only technique can be applied to  
objects with near-linear
union complexity, giving for instance a CF-coloring
of~$O(\log^3 n)$ colors
for pseudodisks using $O(\log n)$ recolorings per update.
\ifShort \else 
This is the first results on semi-dynamic CF-colorings for this
general class of objects. 
\fi
Our fully dynamic solution applies to any
class of object on which weak deletions are possible,
giving for instanct a CF-coloring of~$O(\sqrt{n}\log^2 n)$
colors for points with respect to rectangles at the cost of~$O(\log n)$
recolorings per insertion and~$O(1)$ recolorings per deletion.
\ifShort \else
This constitutes the first fully-dynamic CF-coloring for objects in the plane.
\fi


\begin{thebibliography}{10}

\bibitem{Ajwani2012}
Deepak Ajwani, Khaled Elbassioni, Sathish Govindarajan, and Saurabh Ray.
\newblock Conflict-free coloring for rectangle ranges using $o(n^{0.382})$
  colors.
\newblock {\em Discrete {\&} Computational Geometry}, 48(1):39--52, 2012.

\bibitem{abes-ubulfo-2014}
Boris Aronov, Mark de~Berg, Esther Ezra, and Micha Sharir.
\newblock Improved bounds for the union of locally fat objects in the plane.
\newblock {\em SIAM Journal on Computing}, 43:534--572, 2014.

\bibitem{barnoy-k-degeneracy}
Amotz Bar{-}Noy, Panagiotis Cheilaris, Svetlana Olonetsky, and Shakhar
  Smorodinsky.
\newblock Online conflict-free colouring for hypergraphs.
\newblock {\em Combinatorics, Probability {\&} Computing}, 19(4):493--516,
  2010.

\bibitem{bs-dspsd-80}
Jon~Louis Bentley and James~B. Saxe.
\newblock Decomposable searching problems {I}: Static-to-dynamic
  transformation.
\newblock {\em Journal of Algorithms}, 1(4):301--358, 1980.

\bibitem{cheilaris-scf-14}
Panagiotis Cheilaris, Luisa Gargano, Adele~A. Rescigno, and Shakhar
  Smorodinsky.
\newblock Strong conflict-free coloring for intervals.
\newblock {\em Algorithmica}, 70(4):732--749, 2014.

\bibitem{chen-ocf-06}
Ke~Chen.
\newblock How to play a coloring game against a color-blind adversary.
\newblock In {\em Proceedings of the 22nd {ACM} Symposium on Computational
  Geometry}, pages 44--51, 2006.

\bibitem{chen-ocf-07}
Ke~Chen, Amos Fiat, Haim Kaplan, Meital Levy, Jir{\'{\i}} Matousek, Elchanan
  Mossel, J{\'{a}}nos Pach, Micha Sharir, Shakhar Smorodinsky, Uli Wagner, and
  Emo Welzl.
\newblock Online conflict-free coloring for intervals.
\newblock {\em SIAM Journal on Computing}, 36(5):1342--1359, 2007.

\bibitem{cks-ocfc-09}
Ke~Chen, Haim Kaplan, and Micha Sharir.
\newblock Online conflict-free coloring for halfplanes, congruent disks, and
  axis-parallel rectangles.
\newblock {\em {ACM} Trans. Algorithms}, 5(2):16:1--16:24, 2009.

\bibitem{clrs-ia-09}
Thomas~H. Cormen, Charles~E. Leiserson, Ronald~L. Rivest, and Clifford Stein.
\newblock {\em Introduction to Algorithms}.
\newblock The MIT Press, 3rd edition, 2009.

\bibitem{bcko-cgaa-08}
Mark de~Berg, Otfried Cheong, Marc van Kreveld, and Mark Overmars.
\newblock {\em Computational Geometry: Algorithms and Applications}.
\newblock Springer-Verlag, 3rd ed. edition, 2008.

\bibitem{dynamic-cf-isaac}
Mark de~Berg, Tim Leijsen, Aleksandar Markovic, Andr\'e{} van Renssen, Marcel
  Roeloffzen, and Gerhard Woeginger.
\newblock Dynamic and kinetic conflict-free coloring of intervals with respect
  to points.
\newblock In {\em Proceedings of the 28th International Symposium on Algorithms
  and Computation}, 2017.

\bibitem{dilworth-thm}
Robert~P. Dilworth.
\newblock A decomposition theorem for partially ordered sets.
\newblock {\em Annals of Mathematics}, 51(1):161--166, 1950.

\bibitem{even-cf-03}
Guy Even, Zvi Lotker, Dana Ron, and Shakhar Smorodinsky.
\newblock Conflict-free colorings of simple geometric regions with applications
  to frequency assignment in cellular networks.
\newblock {\em SIAM Journal on Computing}, 33(1):94--136, 2003.

\bibitem{harpeled-cf-05}
Sariel Har{-}Peled and Shakhar Smorodinsky.
\newblock Conflict-free coloring of points and simple regions in the plane.
\newblock {\em Discrete {\&} Computational Geometry}, 34(1):47--70, 2005.

\bibitem{hks-cfcms-10}
Elad Horev, Roi Krakovski, and Shakhar Smorodinsky.
\newblock Conflict-free coloring made stronger.
\newblock In {\em Proceedings ot the 12th Scandinavian Symposium and Workshops
  on Algorithm Theory}, pages 105--117, 2010.

\bibitem{thesis-smorodinsky}
Shakhar Smorodinsky.
\newblock {\em Combinatorial Problems in Computational Geometry}.
\newblock PhD thesis, Tel-Aviv University, 2003.

\bibitem{S-survey-10}
Shakhar Smorodinsky.
\newblock Conflict-free coloring and its applications.
\newblock {\em CoRR}, abs/1005.3616, 2010.

\end{thebibliography}




\end{document}